\documentclass[10pt,english,preprint]{revtex4-1}
\usepackage{libertineRoman}
\usepackage[T1]{fontenc}
\usepackage[latin9]{inputenc}
\usepackage{xcolor}
\usepackage{babel}
\usepackage{mathrsfs}
\usepackage{bm}
\usepackage{amsmath}
\usepackage{amsthm}
\usepackage{amssymb}
\usepackage[unicode=true,pdfusetitle,
 bookmarks=true,bookmarksnumbered=false,bookmarksopen=false,
 breaklinks=true,pdfborder={0 0 0},pdfborderstyle={},backref=false,colorlinks=true]
 {hyperref}
\hypersetup{
 urlcolor=red, citecolor=blue, linkcolor=blue}

\makeatletter
\theoremstyle{plain}
\newtheorem{lem}{\protect\lemmaname}
\theoremstyle{plain}
\newtheorem{thm}{\protect\theoremname}

\makeatother

\providecommand{\lemmaname}{Lemma}
\providecommand{\theoremname}{Theorem}

\begin{document}

\title{Unambiguous discrimination of sequences of quantum states}

\author{Tathagata Gupta}

\affiliation{Physics and Applied Mathematics Unit, Indian Statistical Institute,
203 B. T. Road, Kolkata 700108, India}
\email{tathagatagupta@gmail.com}

\selectlanguage{english}%

\author{Shayeef Murshid}

\affiliation{Electronics and Communication Sciences Unit, Indian Statistical Institute,
203 B. T. Road, Kolkata 700108, India}
\email{shayeef.murshid91@gmail.com}

\selectlanguage{english}%

\author{Somshubhro Bandyopadhyay}

\affiliation{Department of Physical Sciences, Bose Institute, EN 80, Bidhannagar,
Kolkata 700091, India}
\email{som@jcbose.ac.in}

\selectlanguage{english}%
\begin{abstract}
We consider the problem of determining the state of an unknown quantum
sequence without error. The elements of the given sequence are drawn
with equal probability from a known set of linearly independent pure
quantum states with the property that their mutual inner products
are all real and equal. This problem can be posed as an instance of
unambiguous state discrimination where the states correspond to that
of all possible sequences having the same length as the given one.
We calculate the optimum probability by solving the optimality conditions
of a semidefinite program. The optimum value is achievable by measuring
individual members of the sequence, and no collective measurement
is necessary. 
\end{abstract}
\maketitle

\section{Introduction}

One of the remarkable features of quantum theory that shows a radical
departure from classical physics is that distinct quantum states may
not be reliably distinguished from one another. In particular, if
a quantum system is prepared in one of two nonorthogonal states $\left|\psi_{1}\right\rangle $
and $\left|\psi_{2}\right\rangle $, then no quantum measurement could
determine the state of the system with certainty. In other words,
quantum theory allows us to distinguish only between orthogonal states. 

Even though nonorthogonal states cannot be reliably distinguished,
one may still try to glean as much ``which state'' information as
possible. Consider a quantum system prepared in one of several nonorthogonal
states $\left|\psi_{1}\right\rangle ,$$\left|\psi_{2}\right\rangle ,\dots,\left|\psi_{N}\right\rangle $,
but we do not know which one. The objective is to determine, as well
as possible, the state of the system by performing a suitable measurement.
This problem is known as quantum state discrimination (see \citep{Chefles-review-2000,BC-review-2009,Bae-Kwek-review-2015}
for excellent reviews). 

Two approaches are usually considered to study a state discrimination
problem. The first is known as minimum-error discrimination, which
aims to design a measurement that minimizes the average error and
applies to any set of nonorthogonal states. For two states $\left|\psi_{1}\right\rangle $
and $\left|\psi_{2}\right\rangle $ with prior probabilities $p_{1}$
and $p_{2}$, the maximum probability of success is given by $1-p_{e}$,
where 
\begin{alignat}{1}
p_{e} & =\frac{1}{2}\left(1-\sqrt{1-4p_{1}p_{2}\left|\left\langle \psi_{1}\vert\psi_{2}\right\rangle \right|^{2}}\right)\label{error-probability}
\end{alignat}
is the minimum probability of error \citep{Helstrom01975}. 

The second strategy is called unambiguous discrimination, which seeks
definite knowledge of the state balanced against a probability of
failure. Here a measurement outcome either correctly identifies the
given state or is inconclusive, in which case, we do not learn anything
about the state. Once again, for the two-state problem, the maximum
probability of success is given by $1-p_{I}$, where 
\begin{alignat}{1}
p_{I} & =2\sqrt{p_{1}p_{2}}\left|\left\langle \psi_{1}\vert\psi_{2}\right\rangle \right|\label{inconclusive probability}
\end{alignat}
is the minimum probability for an inconclusive result \citep{Ivan-87,Dieks-88,Peres-88,Jaeger-Shimony-95}.
Unlike minimum-error discrimination, which applies to any set of states,
unambiguous discrimination is possible if and only if the given states
are linearly independent \citep{Chefles-98}. Finding optimal solutions,
however, is considerably hard in general (for different approaches
and solutions for specific cases see \citep{Peres-Terno-98,Chefles-Barnett-PLA-98,Jafari+2008,Pang-Wu-2009,Bergou+-2012,Bandyo-14,Sun+-2001,Chefles+2004,Sugimoto+-2010,Eldar-2003,Roa+-2011}). 

In this paper we consider a variant of the state discrimination problem,
namely, sequence discrimination, where, instead of learning about
the state of a single quantum system as in state discrimination, we
wish to do the same about the state of a sequence of quantum systems.
We note that a closely related problem, viz., quantum state comparison,
where the objective is to determine if the members of a given sequence
are all identical or all different, has been studied before \citep{Chefles+2004}.

Sequence discrimination can be described as follows. Suppose that
we are given a sequence of pure quantum states, where each member
(of the given sequence) belongs to a known set of states (this set
will sometimes be referred to as the parent). We do not know the identity
of the individual members but have complete information about the
parent set. The objective is to learn about the given sequence as
well as allowed by quantum theory. As we will explain, for a given
sequence of finite length, this amounts to discriminating between
all sequences of the same length constructed from the parent set.
Since every sequence, by construction, is in a product state, sequence
discrimination is an instance of state discrimination where the concerned
states are all product states. 

As in a state discrimination problem, one could consider either the
minimum-error or unambiguous discrimination strategy for sequence
discrimination. Here we focus on the latter as we wish to identify
the state of a given sequence without error. This would be possible,
as we know from \citep{Chefles-98}, if and only if the states corresponding
to all sequences of the same length form a linearly independent set,
a condition fulfilled if and only if the parent set is linearly independent
\citep{Chefles+2004}. Therefore, one could determine the state of
any given sequence with a nonzero probability, provided the parent
set is linearly independent, and assuming this is indeed the case,
one would then like to know the optimum success probability and the
corresponding measurement. The present paper is about answering these
two questions. 

To solve this problem with full generality, besides the requirement
that the parent set must consist of linearly independent states, one
would set the associated prior probabilities to be different and the
mutual inner products to be unequal and complex. However, in this
work we do not attempt to solve the most general scenario; instead,
we assume that the elements of the parent set are all equally probable
and the inner products are all real and equal and solve the sequence
discrimination problem completely. The optimum success probability
is computed by solving the optimality conditions of a semidefinite
program. We find that the optimum value is achievable by measuring
individual members of the sequence and no collective measurement is
necessary, even though the states under consideration are all product
states. 

One may have noticed that our problem is motivated by the working
of quantum key distribution protocols, especially B92 \citep{B92}
and its generalizations. Recall that in the B92 protocol \citep{B92},
Alice sends a sequence of quantum systems, where each system is prepared
in one of two nonorthogonal pure states, to Bob, who performs an unambiguous
discrimination measurement on each of them to determine its state.
In order to generate the secret key, the conclusive outcomes are kept
and the inconclusive outcomes are discarded. A generalization of this
protocol involves Alice sending a sequence of quantum states, where
each state is chosen from a linearly independent set, so unambiguous
discrimination is possible. Note that the protocol requires Bob to
measure the quantum systems individually as and when he receives them.
However, one could ask whether Bob could do better by performing a
joint measurement on the whole sequence (assuming Bob has access to
quantum memory). Of course, this comes with the drawback that if the
outcome is inconclusive, they will need to discard the entire sequence,
but the possible upshot is that a joint measurement could increase
the probability of identifying the sequence correctly. However, our
result shows that Bob gains no advantage by choosing a collective
measurement over measuring the systems individually. 

The paper is organized as follows. Section \ref{II} discusses the
sequence discrimination problem in detail, presents the necessary
lemmas related to the properties of a collection of pure states with
mutual inner products all being real and equal, and states the main
result as Theorem \ref{main-result}. In Sec. \ref{III} we discuss
the semidefinite programming (SDP) formulation of unambiguous state
discrimination. Section \ref{IV} proves the main result by solving
the optimality conditions of the relevant SDP. This section is divided
into four subsections for easy reading and understanding of the proof.
We conclude the paper with a brief review of the results and a discussion
of the open problems in Sec. \ref{V}. 

\section{Sequence discrimination: Formulation and main result\label{II} }

We begin by describing the general formulation.\textcolor{brown}{{}
}Consider an unknown sequence of $k\in\mathbb{N}$ quantum systems,
each prepared in a state chosen from a known parent set $\left\{ p_{i},\left|\psi_{i}\right\rangle :2\leqslant i\leqslant N\right\} $,
where $p_{i}$ is the prior probability associated with $\left|\psi_{i}\right\rangle $.
The objective is to determine the sequence (more precisely, the state
of the sequence) as well as possible. This can be posed as a state
discrimination problem. 

Let $\left[n\right]=\left\{ 1,2,\dots,n:n\in\mathbb{N}\right\} $
denote the set of natural numbers from $1$ to $n$ and $\mathscr{F}\left(k,N\right)$
be the set of all functions from $\left[k\right]$ to $\left[N\right]$.
Then the state of a sequence is a product state of the form
\begin{alignat*}{1}
\left|\psi_{\sigma}\right\rangle  & =\left|\psi_{\sigma\left(1\right)}\right\rangle \otimes\cdots\otimes\left|\psi_{\sigma\left(k\right)}\right\rangle ,\hspace{1em}\sigma\in\mathscr{F}\left(k,N\right).
\end{alignat*}
To learn about a given sequence of length $k$, we therefore need
to distinguish between all such possible sequences. These sequences
form the set 
\begin{alignat}{1}
\left\{ p_{\sigma},\left|\psi_{\sigma}\right\rangle :\sigma\in\mathscr{F}\left(k,N\right)\right\}  & ,\label{S(N,k)}
\end{alignat}
where $p_{\sigma}=p_{\sigma\left(1\right)}p_{\sigma\left(2\right)}\dots p_{\sigma\left(k\right)}$
is the prior probability associated with the sequence state $\left|\psi_{\sigma}\right\rangle $.
The cardinality of the above set is $N^{k}$. The sequence discrimination
problem is therefore a state discrimination problem involving states
belonging to the set defined by \eqref{S(N,k)}\textcolor{purple}{.} 

Here we consider the problem of unambiguous sequence discrimination.
This requires $\left\{ \left|\psi_{\sigma}\right\rangle \right\} $
to be linearly independent, a condition that is satisfied (for any
$k\geqslant1$) if and only if $\left\{ \left|\psi_{i}\right\rangle \right\} $
is linearly independent \citep{Chefles+2004}. In other words, any
given sequence of unknown pure states can be correctly determined
with nonzero probability if and only if it is composed of states drawn
from a linearly independent set. 

Let us now assume that $\left\{ \left|\psi_{i}\right\rangle \right\} $
is a set of linearly independent states and further assume that they
are equally likely, i.e., $p_{i}=\frac{1}{N}$ for all $i=1,\dots,N$.
This implies that the elements of $\left\{ \left|\psi_{\sigma}\right\rangle \right\} $
are also linearly independent and equally likely with $p_{\sigma}=\frac{1}{N^{k}}$.
Since the elements of $\left\{ \left|\psi_{\sigma}\right\rangle \right\} $
are linearly independent, they can be unambiguously distinguished.
A lower bound on the optimum success probability can be easily obtained. 
\begin{lem}
\label{lower-bound} Let $p$ and $p_{N,k}$ be the respective probabilities
for unambiguous optimal discrimination among the elements of $\left\{ \left|\psi_{i}\right\rangle \right\} $
and $\left\{ \left|\psi_{\sigma}\right\rangle \right\} $. Then
\begin{flalign}
p_{N,k} & \geqslant p^{k}.\label{p(n,K)>=00003Dp^k}
\end{flalign}
\end{lem}
\begin{proof}
First note that every member of a given sequence is an element of
$\left\{ \left|\psi_{i}\right\rangle \right\} $. Let the optimal
measurement that unambiguously distinguishes between the elements
of $\left\{ \left|\psi_{i}\right\rangle \right\} $ be $\mathbb{M}$.
Then, by performing this measurement on individual members of the
sequence, we can determine the state of each of them with probability
$p$. Thus the state of the sequence can be correctly determined with
probability $p^{k}$ (note that for the lower bound to hold the probability
distribution need not be uniform).
\end{proof}
The lower bound in \eqref{p(n,K)>=00003Dp^k} is obtained by the strategy
that unambiguously determines the state of each member of the sequence
separately. However, such a strategy could well be sub-optimal. The
reasoning goes as follows: Since $\left\{ \left|\psi_{\sigma}\right\rangle \right\} $
is a collection of linearly independent product states, to optimally
distinguish between them, a joint measurement on the whole system
may be necessary and, if so, inequality \eqref{p(n,K)>=00003Dp^k}
would be strict. Indeed, there are instances where joint measurements
are required to optimally distinguish between product states (see,
e.g., \citep{Peres-Wootters-1991,Massar-Popescu-1995,Bennett+-1999}).

The main contribution of this paper is to show that if the states
$\left|\psi_{i}\right\rangle $, in addition to being linearly independent,
have the property that their mutual inner products are all real and
equal, then equality holds in \eqref{p(n,K)>=00003Dp^k}. Therefore,
the optimum probability to unambiguously determine the state of an
unknown sequence, whose elements are drawn with equal probability
from a set of linearly independent states with real and equal inner
products, can be achieved by measuring the members of the sequence
individually.

The following results are proved in \citep{Roa+-2011}. The first
gives us the condition under which a collection of pure states, with
inner products real and equal, can be linearly independent. 
\begin{lem}
(from \citep{Roa+-2011}) Let $S_{N}=\left\{ \left|\psi_{i}\right\rangle :2\leqslant i\leqslant N\right\} $
be a set of pure states with the property $\left\langle \psi_{i}\vert\psi_{j}\right\rangle =s\in\mathbb{R}$
for $i\neq j$. The states are linearly independent if and only if
$s\in\left(-\frac{1}{N-1},1\right)$. 
\end{lem}
The lemma tells us that once we require the inner products to all
be real and equal to, say, $s$, then the states $\left|\psi_{i}\right\rangle $
cannot be linearly independent for all permissible values of $s$;
they are linearly independent provided $s\in\left(-\frac{1}{N-1},1\right)$.
The proof follows by requiring the Gram determinant to be greater
than $0$, which is an equivalent criterion for linear independence. 

For a set of linearly independent, equally likely pure states with
real and equal inner products, the following lemma tells us how well
they can be distinguished unambiguously. 
\begin{lem}
\label{p=00003D1-s} (from \citep{Roa+-2011}) Let $S_{N}=\left\{ \left|\psi_{i}\right\rangle :2\leqslant i\leqslant N\right\} $
be a set of equally likely, linearly independent pure states with
the property $\left\langle \psi_{i}\vert\psi_{j}\right\rangle =s$
for $i\neq j$, where $s\in\left(-\frac{1}{N-1},1\right)$. Then the
optimum probability for unambiguous discrimination among the states
$\left|\psi_{i}\right\rangle $ is
\begin{alignat*}{1}
p & =\begin{cases}
1-s, & s\in\left[0,1\right)\\
1+\left(N-1\right)s, & s\in\left(-\frac{1}{N-1},0\right].
\end{cases}
\end{alignat*}
 
\end{lem}
The proof can be found in \citep{Roa+-2011} (the result in \citep{Roa+-2011}
was more general and was proved for states having equal inner products,
real or complex). The basic idea is to attach an ancilla with the
given system (in an unknown state), apply a joint unitary transformation
on the whole system, and finally measure the ancilla in an orthogonal
basis. By choosing an appropriate unitary transformation, the measurement
on the ancilla maps the system of interest onto the unambiguous subspace
with a nonzero probability. 

We now state our main result. 
\begin{thm}
\label{main-result} Let $S_{N,k}=\left\{ \left|\psi_{\sigma}\right\rangle \equiv\left|\psi_{\sigma\left(1\right)}\right\rangle \otimes\cdots\otimes\left|\psi_{\sigma\left(k\right)}\right\rangle :\sigma\in\mathscr{F}\left(k,N\right)\right\} $
be the set of all sequences of $k$ states, where each member of a
sequence is drawn from $S_{N}$ (defined in Lemma \ref{p=00003D1-s})
with equal probability (hence, the sequences are all equiprobable).
Then the optimum probability of unambiguous discrimination between
the elements of $S_{N,k}$ is
\begin{alignat*}{1}
p_{N,k} & =\begin{cases}
\left(1-s\right)^{k}, & s\in\left[0,1\right)\\
\left[1+\left(N-1\right)s\right]^{k}, & s\in\left(-\frac{1}{N-1},0\right].
\end{cases}
\end{alignat*}
 This probability is achievable by measuring the individual systems
forming a sequence. 
\end{thm}
From Lemma \ref{p=00003D1-s} and Theorem \ref{main-result} we see
that $p_{N,k}=p^{k}$ for all $s\in\left(-\frac{1}{N-1},1\right)$.
We will prove this theorem by solving the optimality conditions of
a semidefinite program (SDP). So we proceed by formulating the unambiguous
state discrimination problem as an SDP and deriving the dual problem.

\section{SDP formulation\label{III}}

Given a set of $N$ linearly independent pure states $\left|\chi_{i}\right\rangle $
with prior probabilities $\eta_{i}$, the problem of unambiguous discrimination
can be cast as an SDP \citep{Sugimoto+-2010,Eldar-2003}. The primal
problem is
\begin{equation}
\begin{aligned} & \underset{\bm{p}}{\text{maximize}} &  & \bm{\eta\cdot p}\\
 & \text{subject to} &  & \Gamma-P\succeq0,\\
 &  &  & \bm{p}\succeq0.
\end{aligned}
\label{primal}
\end{equation}
Here $\bm{\eta}=\left(\eta_{1},\dots,\eta_{N}\right)$ and $\bm{p}=\left(p_{1},\dots,p_{N}\right)$,
where $p_{i}$ is the SDP variable representing the probability of
successfully detecting the input $\left|\chi_{i}\right\rangle $;
$\Gamma$ is the Gram matrix whose elements are $\Gamma_{ij}=\left\langle \chi_{i}\vert\chi_{j}\right\rangle $
and $P=\text{diag}\left(p_{1},\dots,p_{N}\right)$. The first constraint
says that the matrix $\Gamma-P$ should be positive semidefinite and
the second constraint is simply the positive semidefiniteness of the
probabilities $p_{i}$. 

To construct the dual SDP, we first construct the Lagrangian 
\begin{alignat*}{1}
L\left(\bm{p},Z,\bm{\bm{z}}\right) & =\bm{\eta\cdot p}+\text{tr}\left[\left(\Gamma-P\right)Z\right]+\bm{\bm{z}}\bm{\cdot}\bm{p},
\end{alignat*}
where the dual variable $Z$ is an $N\times N$ real-symmetric matrix
and $\bm{\bm{z}}$ is a real $N$-tuple. If $Z,\bm{\bm{z}}\succeq0$,
then $L\left(\bm{p},Z,\bm{\bm{z}}\right)\geqslant\bm{\eta\cdot\tilde{p}}$
for any feasible solution $\bm{\tilde{p}}$ of the primal SDP. Therefore,
the inequality must also hold for the optimum $\bm{p}$, say, $\bm{p^{*}}$,
which implies that $L\left(\bm{p^{*}},Z,\bm{\bm{z}}\right)\geqslant\bm{\eta\cdot p^{*}}$.
With this in mind, we define the Lagrange dual function
\begin{alignat*}{1}
g\left(Z,\bm{\bm{z}}\right) & =\sup_{\bm{p}}L\left(\bm{p},Z,\bm{\bm{z}}\right)
\end{alignat*}
and note that it satisfies $g\left(Z,\bm{\bm{z}}\right)\geqslant\max_{\bm{p}}\bm{\eta\cdot p}$.
The dual SDP seeks to 
\[
\begin{aligned} & \underset{Z,\bm{\bm{z}}}{\text{minimize}} &  & g\left(Z,\bm{\bm{z}}\right)\\
 & \text{subject to} &  & Z,\bm{\bm{z}}\succeq0.
\end{aligned}
\]
Consider a family of $N\times N$ matrices $\left\{ F_{i}\right\} $
for $i=1,\dots,N$, where each $F_{i}$ has exactly one nonzero element
$-1$ at position $\left(i,i\right)$. Now note that
\begin{alignat*}{1}
g\left(Z,\bm{\bm{\bm{z}}}\right) & =\sup_{\bm{p}}L\left(\bm{p},Z,\bm{\bm{\bm{z}}}\right)\\
 & =\sup_{\bm{p}}\left\{ \bm{\eta\cdot p}+\text{tr}\left[\left(\Gamma-P\right)Z\right]+\bm{z\cdot p}\right\} \\
 & =\sup_{\bm{p}}\left[\sum_{i=1}^{N}p_{i}\left[z_{i}+\eta_{i}+\text{tr}\left(F_{i}Z\right)\right]+\text{tr}\left(\Gamma Z\right)\right]\\
 & =\left\{ \begin{array}{ccc}
\text{tr}\left(\Gamma Z\right) &  & \text{if\;}z_{i}+\eta_{i}+\text{tr}\left(F_{i}Z\right)=0\,\forall i\\
\infty &  & \text{otherwise}.
\end{array}\right.
\end{alignat*}
Therefore, the dual problem becomes 
\[
\begin{aligned} & \underset{Z,\bm{\bm{\bm{z}}}}{\text{minimize}} &  & \text{tr}\left(\Gamma Z\right)\\
 & \text{subject to} &  & z_{i}+\eta_{i}+\text{tr}\left(F_{i}Z\right)=0,\\
 &  &  & Z,\bm{\bm{\bm{z}}}\succeq0.
\end{aligned}
\]

In the next section we prove the main result. 

\section{Proof of Theorem 1\label{IV}}

\emph{Proof outline}. First note that the primal problem is convex
and there exists a $\bm{p}$ (equivalently $P$) that is strictly
feasible. Under these conditions, Slater's theorem guarantees that
the strong duality holds, and the duality gap is zero. The way we
will proceed is the following. We will present an ansatz $P$ and
obtain a solution for the primal problem, which is not necessarily
optimal. Then we will present candidates for the dual variables $Z$
and $\bm{\bm{\bm{z}}}$ and show that this makes the dual value equal
to the primal one. Since strong duality holds, this implies that our
ansatz must be the optimal solution for the primal problem.

\subsection{Technical Lemmas}

First we will prove a couple of technical lemmas. 
\begin{lem}
\label{block-lemma} If $A$ is a block matrix with each block being
a diagonal matrix of the same size, then $A$ is similar to a block-diagonal
matrix. 
\end{lem}
\begin{proof}
Let $A=\left(\begin{array}{cccc}
A_{11} & A_{12} & \cdots & A_{1n}\\
A_{21} & A_{22} & \cdots & A_{2n}\\
\vdots & \vdots & \ddots & \vdots\\
A_{n1} & A_{n2} & \cdots & A_{nn}
\end{array}\right),$ where the $A_{ij}$'s are diagonal matrices of size $m\times m$.
Let $\alpha_{ijk}$ denote the $k$-th diagonal entry of $A_{ij}$.
Then
\begin{alignat}{1}
A & =\sum_{i,j=1}^{n}\sum_{k=1}^{m}\alpha_{ijk}E_{ij}^{\left(n\right)}\otimes E_{kk}^{\left(m\right)},\label{A}
\end{alignat}
where $E_{\mu\nu}^{\left(t\right)}$ is a $t\times t$ matrix whose
$\left(\mu,\nu\right)$-th entry is $1$ and all other entries are
$0$. Now, if $U$ and $V$ are square matrices, then $U\otimes V=P^{-1}\left(V\otimes U\right)P$
for some permutation matrix $P$ that depends only on the dimensions
of $U$ and $V$ \citep{Henderson-Searle-1981}. It then follows that
$A$ is similar to $\sum_{i,j=1}^{n}\sum_{k=1}^{m}\alpha_{ijk}E_{kk}^{\left(m\right)}\otimes E_{ij}^{\left(n\right)}$,
which has the block-diagonal form $\left(\begin{array}{cccc}
D_{1} & 0 & \cdots & 0\\
0 & D_{2} & \cdots & 0\\
\vdots & \vdots & \ddots & \vdots\\
0 & 0 & \cdots & D_{m}
\end{array}\right)$, where the $\left(i,j\right)$ entry of $D_{k}$ is $\alpha_{ijk}$. 
\end{proof}
\begin{lem}
\label{eigvalues-Lambda} Let $\Lambda$ be a real $n\times n$ matrix
of the form 
\begin{alignat}{1}
\Lambda & =\left(\begin{array}{cccc}
1 & r & \cdots & r\\
r & 1 & \cdots & r\\
\vdots & \vdots & \ddots & \vdots\\
r & r & \cdots & 1
\end{array}\right),\hspace{1em}r\in\mathbb{R}.\label{Lambda}
\end{alignat}
The distinct (except when $r=0$) eigenvalues of $\Lambda$ are $1-r$
and $1+\left(n-1\right)r$.
\end{lem}
\begin{proof}
The characteristic polynomial is
\begin{alignat*}{1}
\det\left(\Lambda-\lambda I\right) & =\det\left(\begin{array}{cccc}
1-\lambda & r & \cdots & r\\
r & 1-\lambda & \cdots & r\\
\vdots & \vdots & \ddots & \vdots\\
r & r & \cdots & 1-\lambda
\end{array}\right)\\
 & =\det\left(\begin{array}{cccc}
h-\lambda & r & \cdots & r\\
h-\lambda & 1-\lambda & \cdots & r\\
\vdots & \vdots & \ddots & \vdots\\
h-\lambda & r & \cdots & 1-\lambda
\end{array}\right),\hspace{1em}\begin{cases}
C_{1}\rightarrow C_{1}+\dots+C_{n}\\
h=1+\left(n-1\right)r
\end{cases}\\
 & =\left(h-\lambda\right)\det\left(\begin{array}{cccc}
1 & r & \cdots & r\\
1 & 1-\lambda & \cdots & r\\
\vdots & \vdots & \ddots & \vdots\\
1 & r & \cdots & 1-\lambda
\end{array}\right)\\
 & =\left(h-\lambda\right)\det\left(\begin{array}{cccc}
1 & r & \cdots & r\\
0 & 1-r-\lambda & \cdots & 0\\
\vdots & \vdots & \ddots & \vdots\\
0 & 0 & \cdots & 1-r-\lambda
\end{array}\right),\hspace{1em}\begin{cases}
R_{i}\rightarrow R_{i}-R_{1}\\
i\neq1
\end{cases}\\
 & =\left(h-\lambda\right)\left(1-r-\lambda\right)^{n-1}.
\end{alignat*}
The distinct (except for $r=0$) eigenvalues are therefore $1-r$
and $h=1+\left(n-1\right)r$.
\end{proof}

\subsection{Eigenvalues of the Gram matrix}

First we would like to calculate the eigenvalues of $\Gamma\left(N,k\right)$,
the Gram matrix of the states of $S_{N,k}$. We begin by finding the
eigenvalues of $\Gamma\left(N,1\right)$, which is a real $N\times N$
matrix 

\begin{alignat}{1}
\Gamma\left(N,1\right) & =\left(\begin{array}{cccc}
1 & s & \cdots & s\\
s & 1 & \cdots & s\\
\vdots & \vdots & \ddots & \vdots\\
s & s & \cdots & 1
\end{array}\right),\hspace{1em}s\in\left(-\frac{1}{N-1},1\right).\label{Gamma(n,1)}
\end{alignat}
The eigenvalues of $\Gamma\left(N,1\right)$ are immediately obtained
by applying Lemma \ref{eigvalues-Lambda}. 
\begin{lem}
\label{eigvalues-Gamma(N,1)} The distinct (except when $s=0$) eigenvalues
of $\Gamma\left(N,1\right)$ are $1-s$ and $1+\left(N-1\right)s$,
where $s\in\left(-\frac{1}{N-1},1\right)$. 
\end{lem}
The proof follows from Lemma \ref{eigvalues-Lambda}.

We will use Lemmas \ref{block-lemma} and \ref{eigvalues-Lambda}
to calculate the eigenvalues of $\Gamma\left(N,k\right)$. 
\begin{thm}
\label{eigvalues-Gamma(N,k)} The eigenvalues of $\Gamma\left(N,k\right)$
are of the form $\left(1-s\right)^{a}\left[1+\left(N-1\right)s\right]^{b}$
for non-negative integers $a$ and $b$ satisfying $a+b=k$. 
\end{thm}
\begin{proof}
We will prove the theorem by induction on $k$. For the proof, we
will need the following result that shows the connection between $\Gamma\left(N,l+1\right)$
and $\Gamma\left(N,l\right)$, where $l\in\mathbb{N}$. 
\begin{lem}
\label{Gamma n, k+1} The Gram matrix of the states of $S_{N,l+1}$
is given by 
\begin{alignat}{1}
\Gamma\left(N,l+1\right) & =\left[\begin{array}{cccc}
\Gamma\left(N,l\right) & s\Gamma\left(N,l\right) & \cdots & s\Gamma\left(N,l\right)\\
s\Gamma\left(N,l\right) & \Gamma\left(N,l\right) & \cdots & s\Gamma\left(N,l\right)\\
\vdots & \vdots & \ddots & \vdots\\
s\Gamma\left(N,l\right) & s\Gamma\left(N,l\right) & \cdots & \Gamma\left(N,l\right)
\end{array}\right]\label{Gamma(n,k+1)}
\end{alignat}
\end{lem}
\begin{proof}
For ease of understanding, denote the elements of $S_{N,l}$ by $\left|\phi_{i}\right\rangle $,
where $i=1,\dots,N^{l}$. Then the elements of $S_{N,l+1}$ are of
the form $\left|\phi_{i}\right\rangle \otimes\left|\psi_{j}\right\rangle $
for $i=1,\dots,N^{l}$ and $j=1,\dots,N$. Then it holds that $\Gamma\left(N,l+1\right)$
must be of the form given by \eqref{Gamma(n,k+1)}, where the $\left(x,y\right)$-th
entry of the $\left(i,j\right)$-th block is the inner product between
$\left|\phi_{x}\right\rangle \otimes\left|\psi_{i}\right\rangle $
and $\left|\phi_{y}\right\rangle \otimes\left|\psi_{j}\right\rangle $.
\end{proof}
From Lemma \ref{eigvalues-Gamma(N,1)}, we know the result holds for
$k=1$. Now assume the result is true for $k=l$. 

Let $G$ be a matrix such that $G\Gamma\left(N,l\right)G^{-1}$ is
diagonal of the form $\alpha=\left(\begin{array}{ccc}
\alpha_{1}\\
 & \ddots\\
 &  & \alpha_{m}
\end{array}\right)$, where $m=N^{l}$ and let $R=\left(\begin{array}{ccc}
G\\
 & \ddots\\
 &  & G
\end{array}\right)$, where the number of $G$ matrices along its diagonal is $N$. Then
\begin{alignat*}{1}
R\Gamma\left(N,l+1\right)R^{-1} & =\left[\begin{array}{cccc}
G\Gamma\left(N,l\right)G^{-1} & sG\Gamma\left(N,l\right)G^{-1} & \cdots & sG\Gamma\left(N,l\right)G^{-1}\\
sG\Gamma\left(N,l\right)G^{-1} & G\Gamma\left(N,l\right)G^{-1} & \cdots & sG\Gamma\left(N,l\right)G^{-1}\\
\vdots & \vdots & \ddots & \vdots\\
sG\Gamma\left(N,l\right)G^{-1} & sG\Gamma\left(N,l\right)G^{-1} & \cdots & G\Gamma\left(N,l\right)G^{-1}
\end{array}\right]\\
 & =\left(\begin{array}{cccc}
\alpha & s\alpha & \cdots & s\alpha\\
s\alpha & \alpha & \cdots & s\alpha\\
\vdots & \vdots & \ddots & \vdots\\
s\alpha & s\alpha & \cdots & \alpha
\end{array}\right),
\end{alignat*}
where we have used $\alpha=G\Gamma\left(N,l\right)G^{-1}$. 

By Lemma \ref{block-lemma}, $\Gamma(n,l+1)$ is therefore similar
to a block-diagonal matrix $\left(\begin{array}{ccc}
D_{1}\\
 & \ddots\\
 &  & D_{m}
\end{array}\right),$ where
\begin{alignat*}{1}
D_{i} & =\left(\begin{array}{cccc}
\alpha_{i} & s\alpha_{i} & \cdots & s\alpha_{i}\\
s\alpha_{i} & \alpha_{i} & \cdots & s\alpha_{i}\\
\vdots & \vdots & \ddots & \vdots\\
s\alpha_{i} & s\alpha_{i} & \cdots & \alpha_{i}
\end{array}\right)=\alpha_{i}\Gamma\left(N,1\right).
\end{alignat*}
Now, by the induction hypothesis,
\begin{alignat*}{1}
\alpha_{i} & =\left(1-s\right)^{a}\left[1+\left(N-1\right)s\right]^{b}
\end{alignat*}
for non-negative integers $a$ and $b$ satisfying $a+b=l$. Thus
the eigenvalues of $D_{i}$ are $\left(1-s\right)^{a+1}\left[1+\left(N-1\right)s\right]^{b}$
or $\left(1-s\right)^{a}\left[1+\left(N-1\right)s\right]^{b+1}$.
Therefore, the result holds for $k=l+1$, proving the theorem.
\end{proof}

\subsection{A feasible solution for the primal problem}

Having found the eigenvalues of $\Gamma\left(N,k\right)$, we will
now guess an ansatz for $\bm{p}$ (equivalently $P$) and then show
that $\Gamma\left(N,k\right)-P\left(N,k\right)$ is positive semidefinite. 
\begin{thm}
\label{Gamma-P} $\Gamma\left(N,k\right)-P\left(N,k\right)$ is positive
semidefinite, where 
\begin{alignat*}{1}
P\left(N,k\right) & =\begin{cases}
\left(1-s\right)^{k}I, & s\in\left[0,1\right)\\
\left[1+\left(N-1\right)s\right]^{k}I, & s\in\left(-\frac{1}{N-1},0\right]
\end{cases}
\end{alignat*}
and $I$ is the $N^{k}\times N^{k}$ identity matrix. 
\end{thm}
\begin{proof}
First note that $\Gamma\left(N,k\right)$ and $P\left(N,k\right)$
are diagonalizable in the same basis as $P\left(N,k\right)$ is a
scalar multiple of the identity. To show that $\left[\Gamma\left(N,k\right)-P\left(N,k\right)\right]$
is positive semidefinite, we use the fact that when two positive semidefinite
matrices $M_{1}$ and $M_{2}$ are diagonalizable in the same basis,
then $M_{1}-M_{2}$ is positive semidefinite if the smallest eigenvalue
of $M_{1}$ is greater than or equal to the largest eigenvalue of
$M_{2}$. 

First consider the case $s\in\left[0,1\right)$. We see that
\begin{alignat*}{1}
\left(1-s\right)^{a}\left[1+\left(N-1\right)s\right]^{b} & \geqslant\left(1-s\right)^{a}\left(1-s\right)^{b}=\left(1-s\right)^{k}
\end{alignat*}
since $\left[1+\left(N-1\right)s\right]^{b}\geqslant\left(1-s\right)^{b}$
as $N\geqslant2$ and $a+b=k$. 

Now consider $s\in\left(-\frac{1}{N-1},0\right]$. Here we have
\begin{alignat*}{1}
\left(1-s\right)^{a}\left[1+\left(N-1\right)s\right]^{b} & \geqslant\left[1+\left(N-1\right)s\right]^{a}\left[1+\left(N-1\right)s\right]^{b}\\
 & =\left[1+\left(N-1\right)s\right]^{k}
\end{alignat*}
since $\left[1+\left(N-1\right)s\right]^{b}\leqslant1$ and $\left(1-s\right)^{a}\geqslant1$
for $s\in\left(-\frac{1}{N-1},0\right]$. 

Therefore, $\Gamma\left(N,k\right)-P\left(N,k\right)$ is positive
semidefinite for $s\in\left(-\frac{1}{N-1},1\right)$. 
\end{proof}
Theorem \ref{Gamma-P} shows that the ansatz
\begin{alignat}{1}
P\left(N,k\right) & =\begin{cases}
\left(1-s\right)^{k}I, & s\in\left[0,1\right)\\
\left[1+\left(N-1\right)s\right]^{k}I, & s\in\left(-\frac{1}{N-1},0\right]
\end{cases}\label{ansatz-P}
\end{alignat}
 will work if there exist positive semidefinite $Z$ and a vector
$\bm{\bm{\bm{\bm{z}}}\succeq0}$ such that
\begin{alignat}{1}
\text{tr}\left[\Gamma\left(N,k\right)Z\right] & =\begin{cases}
\left(1-s\right)^{k} & s\in\left[0,1\right)\\
\left[1+\left(N-1\right)s\right]^{k} & s\in\left(-\frac{1}{N-1},0\right]
\end{cases}\label{Tr(Gamma*Z)}
\end{alignat}
and $z_{i}+\eta_{i}+\text{tr}\left(F_{i}Z\right)=0$ for all $i=1,\dots,N^{k}$. 

\subsection{Optimal solution}

First we show that one can indeed find a suitable $Z$ satisfying
\eqref{Tr(Gamma*Z)}. 
\begin{thm}
\label{main-theorem} For any choice of $N,k\in\mathbb{N}$ there
exists an $N^{k}\times N^{k}$ positive semidefinite matrix $Z\left(N,k\right)$
with diagonal entries $1/N^{k}$ such that
\begin{alignat}{1}
\text{tr}\left[\Gamma\left(N,k\right)Z\left(N,k\right)\right] & =\begin{cases}
\left(1-s\right)^{k}, & s\in\left[0,1\right)\\
\left[1+\left(N-1\right)s\right]^{k}, & s\in\left(-\frac{1}{N-1},0\right].
\end{cases}\label{trace-equality}
\end{alignat}
\end{thm}
\begin{proof}
To prove the theorem we proceed by induction on $k$. 

First consider the case $s\in\left[0,1\right)$. For $k=1$ let 
\begin{alignat*}{1}
Z\left(N,1\right) & =\frac{1}{N}\left(\begin{array}{cccc}
1 & -\frac{1}{\left(N-1\right)} & \cdots & -\frac{1}{\left(N-1\right)}\\
-\frac{1}{\left(N-1\right)} & 1 & \cdots & -\frac{1}{\left(N-1\right)}\\
\vdots & \vdots & \ddots & \vdots\\
-\frac{1}{\left(N-1\right)} & -\frac{1}{\left(N-1\right)} & \cdots & 1
\end{array}\right).
\end{alignat*}
Then by Lemma \ref{eigvalues-Lambda} the eigenvalues of $Z\left(N,1\right)$
are $\frac{1}{N-1}$ and $0$. Hence it is positive semidefinite. 

With the above choice of $Z\left(N,1\right)$ we have 
\begin{alignat*}{1}
\text{tr}\left[\Gamma\left(N,1\right)Z\left(N,1\right)\right] & =\frac{1}{N}\text{tr}\left[\left(\begin{array}{cccc}
1 & s & \cdots & s\\
s & 1 & \cdots & s\\
\vdots & \vdots & \cdots & \vdots\\
s & s & \cdots & 1
\end{array}\right)\left(\begin{array}{cccc}
1 & -\frac{1}{\left(N-1\right)} & \cdots & -\frac{1}{\left(N-1\right)}\\
-\frac{1}{\left(N-1\right)} & 1 & \cdots & -\frac{1}{\left(N-1\right)}\\
\vdots & \vdots & \ddots & \vdots\\
-\frac{1}{\left(N-1\right)} & -\frac{1}{\left(N-1\right)} & \cdots & 1
\end{array}\right)\right]\\
 & =1-s.
\end{alignat*}

Now assume the result holds for $k=l$, i.e., $\text{tr}\left[\Gamma\left(N,l\right)Z\left(N,l\right)\right]=\left(1-s\right)^{l}$
and $Z\left(N,l\right)$ is positive semidefinite. First we will show
that the trace equality holds for $k=l+1$ for $l\geqslant1$. Define
\begin{alignat*}{1}
Z\left(N,l+1\right) & =\frac{1}{N}\left(\begin{array}{cccc}
Z\left(N,l\right) & -\frac{Z\left(N,l\right)}{\left(N-1\right)} & \cdots & -\frac{Z\left(N,l\right)}{\left(N-1\right)}\\
-\frac{Z\left(N,l\right)}{\left(N-1\right)} & Z\left(N,l\right) & \cdots & -\frac{Z\left(N,l\right)}{\left(N-1\right)}\\
\vdots & \vdots & \ddots & \vdots\\
-\frac{Z\left(N,l\right)}{\left(N-1\right)} & -\frac{Z\left(N,l\right)}{\left(N-1\right)} & \cdots & Z\left(N,l\right)
\end{array}\right),\hspace{1em}l\geqslant1.
\end{alignat*}
 Then 
\begin{alignat*}{1}
\text{tr}\left[\Gamma\left(N,l+1\right)Z\left(N,l+1\right)\right] & =\text{tr}\left[\Gamma\left(N,l\right)Z\left(N,l\right)\right]-s\text{tr}\left[\Gamma\left(N,l\right)Z\left(N,l\right)\right]\\
 & =\left(1-s\right)\text{tr}\left[\Gamma\left(N,l\right)Z\left(N,l\right)\right]\\
 & =\left(1-s\right)^{l+1},
\end{alignat*}
which proves the equality holds for $k=l+1$. 

What remains to be shown is that $Z\left(N,l+1\right)$ is positive
semidefinite. The eigenvalues of $Z\left(N,l+1\right)$ are obtained
by applying Theorem \ref{eigvalues-Gamma(N,k)} with $s=-\frac{1}{N-1}$.
The eigenvalues are either $0$ or $\frac{1}{N}\left(\frac{N}{N-1}\right)^{l+1}$,
where the nonzero eigenvalue is obtained for $a=l+1$ and $b=0$.
Therefore, $Z\left(N,k\right)$ is positive semidefinite for $k=l+1$.
Since we have already shown $Z\left(N,1\right)$ is positive semidefinite,
$Z\left(N,k\right)$ is positive semidefinite for all $k\geqslant1$. 

Now consider $s\in\left(-\frac{1}{N-1},0\right]$. For $k=1$ let
\begin{alignat*}{1}
Z\left(N,1\right) & =\frac{1}{N}\left(\begin{array}{cccc}
1 & 1 & \cdots & 1\\
1 & 1 & \cdots & 1\\
\vdots & \vdots & \ddots & \vdots\\
1 & 1 & \cdots & 1
\end{array}\right).
\end{alignat*}
Once again, by applying Lemma \ref{eigvalues-Gamma(N,1)} , we find
that the eigenvalues of $Z\left(N,1\right)$ are $1$ and $0$. Hence
it is positive semidefinite. 

Now with the above choice of $Z\left(N,1\right)$, 
\begin{alignat*}{1}
\text{tr}\left[\Gamma\left(N,1\right)Z\left(N,1\right)\right] & =\frac{1}{N}\text{tr}\left[\left(\begin{array}{cccc}
1 & s & \cdots & s\\
s & 1 & \cdots & s\\
\vdots & \vdots & \cdots & \vdots\\
s & s & \cdots & 1
\end{array}\right)\left(\begin{array}{cccc}
1 & 1 & \cdots & 1\\
1 & 1 & \cdots & 1\\
\vdots & \vdots & \ddots & \vdots\\
1 & 1 & \cdots & 1
\end{array}\right)\right]\\
 & =1+\left(N-1\right)s.
\end{alignat*}
As before, we assume the result holds for $k=l$, i.e. $\text{tr}\left[\Gamma\left(N,l\right)Z\left(N,l\right)\right]=\left[1+\left(N-1\right)s\right]^{l}$
and $Z\left(N,l\right)$ is positive semidefinite. We first show that
the trace equality holds for $k=l+1$ for $l\geqslant1$. Define
\begin{alignat*}{1}
Z\left(N,l+1\right) & =\frac{1}{N}\left(\begin{array}{cccc}
Z\left(N,l\right) & Z\left(N,l\right) & \cdots & Z\left(N,l\right)\\
Z\left(N,l\right) & Z\left(N,l\right) & \cdots & Z\left(N,l\right)\\
\vdots & \vdots & \ddots & \vdots\\
Z\left(N,l\right) & Z\left(N,l\right) & \cdots & Z\left(N,l\right)
\end{array}\right),\hspace{1em}l\geqslant1.
\end{alignat*}
 Then 
\begin{alignat*}{1}
\text{tr}\left[\Gamma\left(N,l+1\right)Z\left(N,l+1\right)\right] & =\text{tr}\left[\Gamma\left(N,l\right)Z\left(N,l\right)\right]+\left(N-1\right)s\text{tr}\left[\Gamma\left(N,l\right)Z\left(N,l\right)\right]\\
 & =\left[1+\left(N-1\right)s\right]\text{tr}\left[\Gamma\left(N,l\right)Z\left(N,l\right)\right]\\
 & =\left[1+\left(N-1\right)s\right]^{l+1}.
\end{alignat*}
Therefore, the trace equality holds for $k=l+1$. What remains to
be shown is that $Z\left(N,l+1\right)$ is positive semidefinite.
Now note that $Z\left(N,l+1\right)$ is an $N^{l+1}\times N^{l+1}$
matrix whose all elements are $\frac{1}{N^{l+1}}$. Therefore, its
eigenvalues are $0$ and $1$; hence, $Z\left(N,l+1\right)$ is positive
semidefinite. Since we have already shown $Z\left(N,1\right)$ is
positive semidefinite, this completes the proof. 

Thus we have proved the existence of positive semidefinite $Z\left(N,k\right)$
that satisfies the trace equality \eqref{trace-equality} for all
$s\in\left(-\frac{1}{N-1},1\right)$. 
\end{proof}
Since $\eta_{i}=1/N^{k}$ for all $i$, by choosing $\bm{\bm{\bm{\bm{z}}}}=\bm{0}$
(null vector), one has $z_{i}+\eta_{i}+\text{tr}\left(F_{i}Z\right)=0$
for all $i$. This completes the proof of our main result, Theorem
\ref{main-result}. 

Let us briefly go through the key elements of the proof once again.
The proof was based on guessing an appropriate primal variable $\bm{p}$
(equivalently $P$) and showing its optimality. Since the primal problem
is convex with a nonempty feasible set, strong duality holds. We showed
that there exist feasible dual variables $Z$ and $\bm{\bm{\bm{z}}}$
such that $\text{tr}\left(\Gamma Z\right)=\bm{\eta\cdot p}$, which
is the dual objective function; hence, our guessed $\bm{p}$ is the
optimal solution.

\section{Conclusions\label{V}}

We considered the problem of unambiguously determining the state $\left|\psi_{\sigma}\right\rangle $
of an unknown quantum sequence of length $k\geqslant1$, where the
elements of the given sequence are drawn with equal probability from
a set of linearly independent pure states $S_{N}=\left\{ \left|\psi_{i}\right\rangle :2\leqslant i\leqslant N\right\} $
with real and equal inner products. This (and even the most general
one without any assumption about inner product and/or prior probabilities)
can be posed as an unambiguous state discrimination problem, where
the objective is to discriminate between the states of all such possible
sequences. 

Let $S_{N,k}=\left\{ \left|\psi_{\sigma}\right\rangle \right\} $
be the set of all possible sequences of length $k$. Let $p$ and
$p_{N,k}$ be the optimum probabilities for unambiguously discriminating
between the elements of $S_{N}$ and $S_{N,k}$, respectively. A simple
argument shows that $p_{N,k}\geqslant p^{k}$, where the lower bound
is achievable by measuring individual members of the sequence. Since
any sequence of length $k$ is a composite quantum system comprising
$k$ subsystems, one might expect the inequality, in general, to be
strict, i.e., $p_{N,k}>p^{k}$, and to achieve the optimum value,
a joint measurement is required. 

Following earlier works on unambiguous state discrimination \citep{Sugimoto+-2010,Eldar-2003},
we formulated the sequence discrimination problem as an SDP and calculated
the optimum probability by solving the optimality conditions. In particular,
we showed that $p_{N,k}=p^{k}$; thus, the optimum value is achieved
by performing measurements on the individual members of the sequence.

Several problems in this context are still left open. First is where
the inner products of the states $\left|\psi_{i}\right\rangle $ belonging
to the parent set $S_{N}$ are equal but complex. For $k=1$, this
reduces to the standard unambiguous state discrimination problem which
has been solved completely \citep{Roa+-2011}. However, we could not
solve the sequence discrimination problem in this scenario using the
same approach. 

Second is a more general scenario where $S_{N}$ is simply a set of
linearly independent pure states without any restrictions on the inner
products. In this case, we carried out thousands of numerical SDP
experiments with a limited number of parent states and very short
sequences, namely, $N=3$ and $k=2,3$, and the results (assuming
uniform prior probabilities) seemed to suggest that the optimum value,
once again, is achievable by measuring the individual members without
requiring any joint measurement. 

In our scenario, as well as in more general ones (which we could not
solve), repetitions of states are allowed in a sequence. The third
problem considers the situation where it is not. Restrict $k<N$ and
by $S'_{N,k}$ denote the set of sequences of length $k$, where no
element is repeated. If we let $\mathscr{G}(k,N)$ be the set of injective
functions from $\left[k\right]$ to $\left[N\right]$, then 
\begin{alignat*}{1}
S_{N,k}^{\prime} & =\left\{ \left|\psi_{\tau(1)}\right\rangle \otimes\cdots\otimes\left|\psi_{\tau(k)}\right\rangle :\tau\in\mathscr{G}(k,N)\right\} 
\end{alignat*}
The cardinality of this set is ${\displaystyle ^{N}P_{k}=\frac{N!}{(N-k)!}}$
(note that the injective functions from $[k]$ to $[N]$ for $k\leqslant N$
correspond to the permutation of $k$ objects chosen from $N$ objects).
Now assume that the inner products of the states $\left|\psi_{i}\right\rangle $
are equal and positive, say, $s>0$. Then $S_{N,k}^{\prime}\subset S_{N,k}$,
where $S_{N,k}$ is the set of sequences considered in this paper.
Under these restricted conditions, numerical experiments ($N=3$,
$k=3$) still suggest that the optimal probability for distinguishing
between the elements of $S_{N,k}^{\prime}$ unambiguously once again
obeys $(1-s)^{k}$. However, for arbitrary values of both $N$ and
$k$, whether one could benefit from collective measurements in this
scenario is an interesting problem to consider in the future. 

To summarize, sequence discrimination considers the problem of distinguishing
between sequences of some fixed length whose members are drawn from
a set of pure states, the parent set. In an unambiguous sequence discrimination
problem, the parent set must consist of linearly independent states;
otherwise unambiguous discrimination will not be possible. We showed
that if the elements of the parent set have the property that the
inner products are all real and equal, then optimal unambiguous sequence
discrimination does not require collective measurements and measuring
the individual members will suffice. However, whether collective measurements
would be necessary for general scenarios, without assumptions about
inner products or prior probabilities, remains open, and so far our
numerical attempts have failed to yield a counter example.

\end{document}